\documentclass[journal,10pt]{IEEEtran}
\usepackage{amsmath}
\usepackage{amsthm}   
\usepackage{amssymb}
\usepackage{cite}
\usepackage{graphicx}
\usepackage{subfigure}
\usepackage{epstopdf} 
\usepackage[square, comma, sort&compress, numbers]{natbib}
\usepackage{color} 
\newtheorem{prop}{Proposition}
\usepackage{algorithm} 
\usepackage{algorithmic}
\newcommand{\figref}[1]{Fig.~\ref{#1}}

\begin{document}

\title{Secure Transmission for Intelligent Reflecting Surface-Assisted mmWave and Terahertz Systems
}
\author{\IEEEauthorblockN{ 
Jingping~Qiao,~\IEEEmembership{Member,~IEEE,}
and Mohamed-Slim Alouini,~\IEEEmembership{Fellow,~IEEE} 
}
\thanks{J. Qiao is with Shandong Normal University, Shandong Province, China, and also with the Computer Electrical, and Mathematical Science and Engineering Division, King Abdullah University of Science and Technology, Thuwal 23955-6900, Saudi Arabia (e-mail: jingpingqiao@sdnu.edu.cn).

 M.-S. Alouini is with the Computer Electrical, and Mathematical Science and Engineering Division, King Abdullah University of Science and Technology (KAUST), Thuwal 23955-6900, Saudi Arabia (e-mail: slim.alouini @kaust.edu.sa).}
}

\maketitle
\begin{abstract}
This letter focuses on the secure transmission for an intelligent reflecting surface (IRS)-assisted millimeter-wave (mmWave) and terahertz (THz) systems, 
in which a base station (BS) communicates with its destination via an IRS, in the presence of a passive eavesdropper.
To maximize the system secrecy rate, the transmit beamforming at the BS and the reflecting matrix at the IRS are jointly optimized with transmit power and discrete phase-shift constraints. 
It is first proved that the beamforming design is independent of the phase shift design under the rank-one channel assumption. 
The formulated non-convex problem is then converted into two subproblems, which are solved alternatively.
Specifically, the closed-form solution of transmit beamforming at the BS is derived, and the semidefinite programming (SDP)-based method and element-wise block coordinate descent (BCD)-based method are proposed to design the reflecting matrix.
The complexity of our proposed methods is analyzed theoretically.
Simulation results reveal that the proposed IRS-assisted secure strategy can significantly boost the secrecy rate performance, regardless of eavesdropper's locations (near or blocking the confidential beam).

\end{abstract}

\section{Introduction}
\IEEEPARstart{T}{erahertz} (THz) and millimeter-wave (mmWave) are the key technologies for beyond 5G and next generation (6G) wireless communications, which can enable ultra-high data-rate communications and improve date security \cite{mmwaveSurvey,NatureShuping}. 
However, the severe propagation loss and blockage-prone nature in such high-frequency bands also pose greater challenges to information security, such as short secure propagation distance and unreliable secure communications \cite{Ma2018Nature,Ela2020Open}. 
These phenomena are more serious for THz bands due to higher frequency. Thus a new approach for information secure transmission in high frequency bands is urgently needed. 


Recently, the intelligent reflecting surface (IRS) \cite{Renzo2019J,Nad2020open,Nad2020TWC}
has emerged as an invaluable way for widening signal coverage and overcome high pathloss of mmWave and THz systems \cite{Wang201908arXiv,Chen2019ICCC,Pan2020TWC} and has drawn increasing attention in secure communications. 
In IRS-assisted secure systems, the IRS intelligently adjusts its phase shifts to steer the signal power to desired user, and reduce information leakage \cite{RuiZhang2019IEEEWCL}.
To maximize the secrecy rate, the active transmit beamforming and passive reflecting beamforming were jointly designed in \cite{Shen2019CL,Dong2020WCL,Robert2019GLOBECOM}. 
However, the above works mainly focus on microwave systems, while IRS-assisted secure mmWave/THz systems still remain unexplored. 
Moreover, extremely narrow beams of mmWave/THz waves can cause serious information leakage of beam misalignment and the costly implementation of continuous phase  control. 
Besides, the blockage-prone nature of high frequency bands may lead to a serious secrecy performance loss, since eavesdroppers can not only intercept but also block legal communications.

Motivated by the aforementioned problems, this letter investigates the IRS-assisted secure transmission in mmWave/THz bands. 
Under the discrete phase-shift assumption, a joint optimization problem of transmit beamforming and reflecting matrix is formulated to maximize the secrecy rate. It is proved that under the rank-one channel model, the transmit beamforming design is independent of reflecting matrix design. 
Thus the formulated non-convex problem is solved by converting it into two subproblems. 
The closed-form transmit beamforming is derived and the the hybrid beamforming structure is designed adopting orthogonal matching pursuit (OMP) method. 
Meanwhile, the semidefinite programming (SDP)-based method and the element-wise block coordinate descent (BCD) method are proposed to obtain the optimal discrete phase shifts.
Simulation results demonstrate that the proposed methods can achieve near-optimal secrecy rate performance 
with discrete phase shifts.


\section{System Model and Problem Formulation}
\begin{figure}[h]
\centering
\includegraphics[width=0.45\textwidth]
{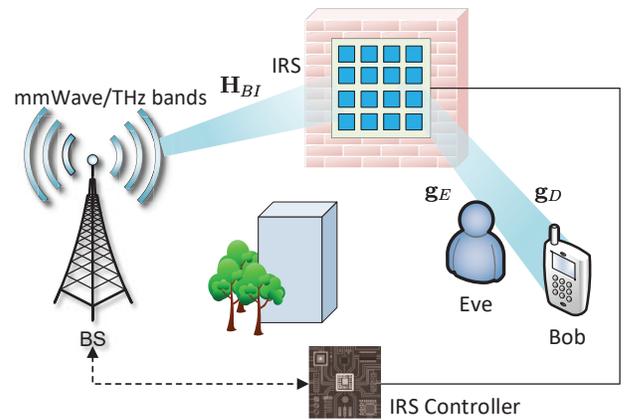}
\caption{System model for IRS-assisted secure mmWave/THz system, where BS communicates with desired user Bob, in the presence of an eavesdropper.}
\label{fig_1}
\end{figure}
The IRS-assisted secure mmWave/THz system is considered in this letter. As shown in \figref{fig_1}, one base station (BS) with $M$ antennas communicates with a single-antenna user Bob in the presence of a single-antenna eavesdropper Eve, which is an active user located near Bob. 
To protect confidential signals from eavesdropping, an IRS with a smart controller is adapted to help secure transmission.   
Note that due to deep path loss or obstacle blockage, there is no direct link between source and destination or eavesdropper.  

\subsection{Channel Model}
It is assumed that the IRS with $N$ reflecting elements is installed on some high-rise buildings around the desired receiver Bob.  
Thus, the LoS path is dominant for the BS-IRS channel and the rank-one channel model is adopted 
\begin{equation}
\mathbf{H}_{BI}^{H}=\sqrt{MN}\alpha_B G_{r}G_{t}\mathbf{a}\mathbf{b}^{H},
\end{equation}
where $\alpha_B$ is the complex channel gain \cite{Bar2017TVT}, and $G_{r}$ and $G_{t}$ are receive and transmit antenna gain.
$\mathbf{a}\in \mathbb{C}^{N\times 1}$ and $\mathbf{b}\in \mathbb{C}^{M\times 1}$ denote the array steering vector at IRS and BS, respectively. 

The IRS-Bob/Eve channel is assumed as\footnote{All channels are assumed to be perfectly known at BS, and the results derived can be considered as the performance upper bound. } 
\begin{equation}
\mathbf{g}_k=\sqrt{\frac{N}{L}}\sum^L_{i=1}\alpha_{k,i}G_{r}^kG_{I}^{k}\mathbf{a}_{k,I},
\end{equation}
where $k=\{D,E\}$, $L$ is the number of paths from IRS to $k$ node, $\mathbf{a}_{k,I}$ denotes the transmit array steering vector at IRS.
\subsection{Signal Model}
In the IRS-assisted secure mmWave/THz system, the BS transmits signal $s$ with power $P_s$ to an IRS, and the IRS adjusts phase shifts of each reflecting element to help reflect signal to Bob. 
We assume $\mathbf{\Theta}\!=\!\text{diag}\{ \!e^{j\theta_{1}}\!, e^{j\theta_{2}}\!,\dots,e^{j\theta_{N}}\!\}$ as the reflecting matrix, where $\theta_{i}$ is the phase shift of each element. 
Different from traditional IRS-assisted strategies with continuous phase shifts, the discrete phase shifts are considered in terms of IRS hardware implementation.
Specifically, $\theta_{i}$ can only be chosen from a finite set of discrete values $\mathcal{F}\!=\!\{0,\Delta\theta,..., (L_P\!-\!1)\Delta\theta\}$, $L_P$ is the number of discrete values, and $\Delta\theta=2\pi/L_P$.

The received signal at user Bob can be expressed as 
\begin{equation}\label{eq_3}
y_D=\mathbf{g}_{D}^{H}\mathbf{\Theta}\mathbf{H}_{BI}^{H}\mathbf{w}s+n_D,
\end{equation}
where $\mathbb{E}\{|s|^2\}\!\!=\!\!1$, 
$n_D\!\sim\!\mathcal{CN}(0, \sigma_D^2)$ is the noise at destination.
$\mathbf{w}=\mathbf{F}_{RF}\mathbf{f}_{BB}$ is the transmit beamforming at the BS, where $\mathbf{F}_{RF}\!\in\!\mathbb{C}^{M\times R}$ is analog beamformer and $\mathbf{f}_{BB}\!\in\! \mathbb{C}^{R\times 1}$ is digital beamformer, $R$ is the number of radio frequency (RF) chains.
Similarly, the received signal at eavesdropper is written as
\begin{equation}\label{eq_4}
y_E=\mathbf{g}_E^{H}\mathbf{\Theta}\mathbf{H}_{BI}^{H}\mathbf{w}s+n_E,
\end{equation}
where $n_E \sim~\mathcal{CN}(0, \sigma_E^2)$ denotes the noise at eavesdropper.

Then the system secrecy rate can be written as
\begin{equation}\label{eq_Rs}
R_s=\left[\log_2\left(\frac{1+\frac{1}{\sigma_D^2}|\mathbf{g}_{D}^{H}\mathbf{\Theta}\mathbf{H}_{BI}^{H}\mathbf{w}|^2}
{1+\frac{1}{\sigma_E^2}|\mathbf{g}_{E}^{H}\mathbf{\Theta}\mathbf{H}_{BI}\mathbf{w}|^2}\right)\right]^+,
\end{equation}
where $[x]^+=\max\{0,x\}$.

\vspace{-0.2cm}
\subsection{Problem Formulation}
To maximize the system secrecy rate, the joint optimization problem of transmit beamforming and reflecting matrix is formulated as
\begin{equation*}
\setlength{\abovedisplayskip}{2pt}
\setlength{\belowdisplayskip}{3pt}
(\mathrm{P1}):~
\max_{\mathbf{w},\mathbf{\Theta}} 
\frac{\sigma_D^2+|\mathbf{g}_{D}^{H}\mathbf{\Theta}\mathbf{H}_{BI}^{H}\mathbf{w}|^2}
{\sigma_E^{2}+|\mathbf{g}_{E}^{H}\mathbf{\Theta}\mathbf{H}_{BI}^{H}\mathbf{w}|^2}
\end{equation*}
\begin{equation}
\begin{array}{ll} 
\text{s.t.} &\|\mathbf{w}\|^2\leq P_s,\\
& \theta_i\in \mathcal {F}, \forall i.\\
\end{array}
\end{equation}
The formulated problem is non-convex and is quite challenging to be solved directly, because of coupled variables $(\mathbf{w},\mathbf{\Theta})$ and the non-convex constraint of $\theta_i$. 
To cope with this difficulty, the original problem $\mathrm{P1}$ is converted into two subproblems, which can be solved alternatively. 

\section{Secrecy Rate Maximization} 
To find out solutions, we propose to convert problem $\mathrm{P1}$ into two subproblems. This idea is based on the fact that two subproblems are independent of each other. 
Specifically, the closed-form solution of $\mathbf{w}$ will be first derived. Then the SDP-based method and element-wise BCD method will be proposed to obtain the solution of reflecting matrix.

\vspace{-0.2cm}
\subsection{Transmit Beamforming Design}

Since the rank-one channel is assumed from BS to IRS, 
the subproblem with respect to beamformer $\mathbf{w}$ is expressed as
\begin{align}
(\mathrm{P2.1}):~&
\max_{\mathbf{w}} 
\frac{\sigma_D^2+|\alpha_{B}G_{r}G_{t}\mathbf{g}_{D}^{H}\mathbf{\Theta}\mathbf{a}|^2|\mathbf{b}^{H}\mathbf{w}|^2}
{\sigma_E^{2}+|\alpha_{B}G_{r}G_{t}\mathbf{g}_{E}^{H}\mathbf{\Theta}\mathbf{a}|^2|\mathbf{b}^{H}\mathbf{w}|^2}\nonumber\\
&~\text{s.t.} ~~\|\mathbf{w}\|^2\leq P_s.
\end{align}
\begin{prop}\label{prop_1}
\setlength{\abovedisplayskip}{3pt}
\setlength{\belowdisplayskip}{3pt}
Under the positive secrecy rate constraint, the suboptimal problem of $\mathbf{w}$ is equivalent to 
\begin{align}
(\mathrm{P2.1}^{'}):~
\max_{\mathbf{w}} 
|\mathbf{b}^{H}\mathbf{w}|^2,\nonumber
~~~\text{s.t.} ~\|\mathbf{w}\|^2\leq P_s.
\end{align}
It is independent of reflecting matrix design.
For any value of $\mathbf{\Theta}$, the transmit beamformer solution is $\mathbf{w}^{opt}=\sqrt{P_s}\frac{\mathbf{b}}{\|\mathbf{b}\|}$. 
\end{prop}
\begin{proof}
See Appendix A.
\end{proof}
After obtaining $\mathbf{w}^{opt}$ from Proposition 1, $(\mathbf{F}_{RF}^{opt}, \mathbf{f}_{BB}^{opt})$ with a full-connected architecture can be easily derived by typical hybrid precoding methods, such as OMP algorithm \cite{Tropp2007TIT}.

\subsection{Reflecting Matrix Design}
To facilitate the following mathematical operations, we first define $\hat{\pmb{\theta}}=[e^{j\theta_1},e^{j\theta_2},...,
e^{j\theta_N}]^{T}$, then we have $\mathbf{\Theta}=\text{diag}\{\hat{\pmb{\theta}}\}$.    
The subproblem of reflecting matrix can be rewritten as 
\vspace{-0.2cm}
\begin{align}
(\mathrm{P2.2}):~&
\max_{\hat{\pmb{\theta}}} 
\frac{1+\frac{1}{\sigma_D^2}|\hat{\pmb{\theta}}^{T}\text{diag}\{\mathbf{g}_{D}^{H}\}\mathbf{H}_{BI}^{H}\mathbf{w}|^2}
{1+\frac{1}{\sigma_E^2}|\hat{\pmb{\theta}}^{T}\text{diag}\{\mathbf{g}_{E}^{H}\}\mathbf{H}_{BI}^{H}\mathbf{w}|^2},\nonumber\\
&\text{s.t.} ~ \hat{\pmb{\theta}}=[e^{j\theta_1},e^{j\theta_2},...,e^{j\theta_N}]^{T},~
\theta_i\in \mathcal {F}, \forall i.\nonumber
\end{align}\label{P_2.2}
It is obvious that the variable $\theta_i$ only takes a finite number of values from $\mathcal{F}$. 
Thus, problem $\mathrm{P2.2}$ is feasible to be solved with an exhaustive search method. 
However, due to a large feasible set of $\hat{\pmb{\theta}}$ ($N^{L_P}$ possibilities), the complexity of such method is considerably high. 
To cope with this, the SDP-based algorithm and element-wise BCD algorithm are proposed.
  
\subsubsection{SDP-based Algorithm} 
Define $\mathbf{\Phi}=\hat{\pmb{\theta}}^{*}(\hat{\pmb{\theta}}^{*})^{H}$ and relax discrete variables $\theta_i$ into continuous $\theta_i\in[0,2\pi)$, i.e., $|e^{j\theta_i}|=1, \forall i$, then we can rewrite problem $\mathrm{P2.2}$ as $\max_{\mathbf{\Phi}\succeq 0} \frac{\text{tr}(\hat{\mathbf{R}}_{RD}\mathbf{\Phi})}{\text{tr}(\hat{\mathbf{R}}_{RE}\mathbf{\Phi})}$ with rank-one constraint, $\text{rank}(\mathbf{\Phi})=1$.

Since $\text{rank}(\mathbf{\Phi})=1$ is a non-convex constraint, the semidefinite relaxation (SDR) is adopted to relax this constraint. 
Using Charnes-Cooper transformation approach, we define $\mathbf{X}=\mu\mathbf{\Phi}$  and  $\mu=1/\text{tr}(\hat{\mathbf{R}}_{RE}\mathbf{\Phi})$. Then problem $\mathrm{P2.2}$ is rewritten as
\begin{align}\label{eq_11}
(\mathrm{P2.2^{'}}):~& \max_{\mu\geq 0,\mathbf{X}\succeq 0} 
\text{tr}(\hat{\mathbf{R}}_{RD}\mathbf{X})  \nonumber\\
&\begin{array}{ll} 
\text{s.t.}& 
\text{tr}(\hat{\mathbf{R}}_{RE}\mathbf{X})=1,\\
&\text{tr}(\mathbf{E}_{n}\mathbf{X})=\mu, \forall n.\\
\end{array}
\end{align}
where $\hat{\mathbf{R}}_{RD}=\frac{1}{N}\mathbf{I}_{N}+\frac{1}{\sigma_{D}^2}\text{diag}\{\mathbf{g}_{D}^{*}\}
\mathbf{H}_{BI}^{H}\mathbf{w}\mathbf{w}^{H}\mathbf{H}_{BI}
\text{diag}\{\mathbf{g}_{D}\}$,
 $\hat{\mathbf{R}}_{RE}\!=\!\!\frac{1}{N}\mathbf{I}_{N}\!+\!\frac{1}{\sigma_{E}^2}\text{diag}\{\mathbf{g}_{E}^{*}\}
\mathbf{H}_{BI}^{H}\mathbf{w}\mathbf{w}^{H}\mathbf{H}_{BI}
\text{diag}\{\mathbf{g}_{E}\}$, and $\mathbf{E}_n$ means the element on position $(n,n)$ is $1$ and 0 otherwise.

Problem $\mathrm{P2.2^{'}}$ is a standard SDP problem, and can be  solved by adopting Interior-point method or CVX tools. \cite{QingWu2019TCOM}
Then the rank-one solution $\hat{\pmb{\theta}}^{opt}\!\!=\![
e^{j\tilde{\theta}_1^{opt}},e^{j\tilde{\theta}_2^{opt}},
 ...,e^{j\tilde{\theta}_N^{opt}}]^{T}$ can be achieved by the Gaussian randomization method. 
To obtain the discrete solution of problem $\mathrm{P2.2}$, we quantify the continuous solution $\hat{\pmb{\theta}}^{opt}$ as the  nearest discrete value in set $\mathcal{F}$, and the following principle is adopted
\begin{equation}\label{eq_12}
\theta_{i}^{opt}=\arg\min_{\theta_i\in \mathcal{F}} |e^{j\tilde{\theta}_i^{opt}}-e^{j\theta_i}|, ~~\forall i.
\end{equation}
Note that since the closed-form solution of $\theta_i$ is not obtained, the entire process of SDP-based method needs to be done for each transmission block, which leads to high complexity.

\subsubsection{Element-Wise BCD Algorithm} 
To obtain the closed-form solution of reflecting matrix, the element-wise BCD method \cite{Robert2019GLOBECOM} is employed in this section.
Taking each $\theta_i$ as one block in the BCD, we can iteratively derive the continuous solutions of phase shifts using Proposition \ref{prop_2}. 
 
\begin{prop}\label{prop_2}
There exists one and only one optimal $\theta_i^{opt}$ to maximize the secrecy rate, and
\begin{equation}\label{eq_18}
\small
\setlength{\abovedisplayskip}{3pt}
\setlength{\belowdisplayskip}{3pt}
\theta_i^{opt}\!\!=\!\!
\left\{\begin{array}{ll} 
\!\!\tilde{\theta}_i^{opt},\!\!&\!\!\! c_{D,i}d_{E,i}\cos(p_{E,i})\!<\!c_{E,i}d_{D,i}\!\cos(p_{D,i})\\
\!\!\tilde{\theta}_i^{opt}\!+\!\pi,\!\!&\!\!\! \text{otherwise} 
\!\!\end{array}\right.
\end{equation}
where $\tilde{\theta}_i^{opt}$ is shown in (\ref{eq_13}) 
on the top of next page, and 
\newcounter{mytempeqncnt}
\begin{figure*}[!t]
\setcounter{mytempeqncnt}{\value{equation}}
\setcounter{equation}{12}
\begin{equation}\label{eq_13}
\tilde{\theta}_i^{opt}=-\arctan\!\left(\frac{c_{D,i}d_{E,i}\sin(p_{E,i})\!-\!c_{E,i}d_{D,i}\sin(p_{D,i})}{c_{D,i}d_{E,i}\cos(p_{E,i})\!-\!c_{E,i}d_{D,i}\cos(p_{D,i})}\right)
-\arcsin\left(\!\frac{-d_{D,i}d_{E,i}\sin(p_{E,i}-p_{D,i})}{\sqrt{c_{D,i}^{2}d_{E,i}^{2}+\!c_{E,i}^{2}d_{D,i}^{2}
\!-\!2c_{D,i}c_{E,i}d_{D,i}d_{E,i}\cos(p_{E,i}\!-\!p_{D,i})}}\!\right)
\end{equation}
\setcounter{equation}{\value{mytempeqncnt}}
\hrulefill
\vspace*{-15pt} %
\end{figure*}
\setcounter{equation}{13}
\begin{align}
&c_{k,i}=1+\frac{1}{\sigma_k^2}\left|g_{k,i}^{*}\mathbf{h}_{BI,i}^{H}\mathbf{w}\right|^2+\frac{1}{\sigma_k^2}\left|\sum_{m\neq i}e^{j\theta_{m}}g_{k,m}^{*}\mathbf{h}_{BI,m}^{H}\mathbf{w}\right|^2,\nonumber\\
&d_{k,i}=\frac{2}{\sigma_k^2}\left|\left(g_{k,i}^{*}\mathbf{h}_{BI,i}^{H}\mathbf{w}\right)
\cdot
\left(\sum_{m\neq i}e^{-j\theta_{m}}g_{k,m}\mathbf{w}^{H}\mathbf{h}_{BI,m}\right)\right|,\nonumber\\
&p_{k,i}=\angle\left(g_{k,i}^{*}\mathbf{h}_{BI,i}^{H}\mathbf{w}\sum_{m\neq i}
e^{-j\theta_{m}}g_{k,m}\mathbf{w}^{H}\mathbf{h}_{BI,m}\right).\nonumber
\end{align}
in which $k\in\{D,E\}$. 
\end{prop}
\begin{proof}
See Appendix B.
\end{proof}
With the above proposition, the optimal discrete solution of $\theta_i$ can be chosen using the same quantization principle in (\ref{eq_12}), and the entire element-wise BCD-based secrecy rate maximization algorithm can be summarized as Algorithm \ref{alg_1}.
Since the objective function in $\mathrm{P2.2}$ is non-decreasing after each iteration and upper-bounded by a finite value of generalized eigenvalue problem \cite{Qiao2018TVT},  the convergence is guaranteed.

\begin{algorithm}
  \caption{E-BCD based Secrecy Rate Maximization
  }
  \label{alg_1}
  \begin{algorithmic}[1]
  
  \STATE 
  Initialize  
  $\mathbf{\Theta}^{0}=\text{diag}\{e^{j\theta_{1}^{0}},e^{j\theta_{2}^{0}},...,e^{j\theta_{N}^{0}}\}$,  $\varepsilon$, and set $n=0$.
  \STATE Find $\mathbf{w}^{opt}$ using Proposition \ref{prop_1}, then calculate optimal $\mathbf{F}_{BF}^{opt}$ and  $\mathbf{f}_{BB}^{opt}$ using OMP method.
  \REPEAT  
  \STATE $n=n+1$
  
  \STATE
\textbf{for} $i=1,2,...,N$ 
\textbf{do}
 \STATE calculate $\theta_{i}^{n}$ using Proposition 2, and quantify it as discrete solution.  
\STATE
\textbf{end for}
  \UNTIL 
  {$\|\mathbf{\Theta}^{n}-\mathbf{\Theta}^{n-1}\|\leq \varepsilon$}.
  
  \RETURN $(\mathbf{w}^{opt})^{'}=\mathbf{F}_{BF}^{opt}\mathbf{f}_{BB}^{opt}$ and  $\mathbf{\Theta}^{opt}=\mathbf{\Theta}^{n}$
  \end{algorithmic}
\end{algorithm}

\subsection{Complexity Analysis}
The total complexity of the SDP-based method is about $\mathcal{O}(N_{gaus}N^8)$, 
which is mainly determined by the complexity of solving SDP problem and the number of rank-one solutions to construct the feasible set of Gaussian randomization method `$N_{gaus}$'. 
While the complexity of the proposed element-wise BCD method is about $\mathcal{O}(N(NM\!+\!L_P)N_{iter})$, which relies on the complexity of $\theta_i$ calculation and the iteration number $N_{iter}$ for $\mathbf{\Theta}^{opt}$. 
Intuitively, these methods have lower complexity than exhaustive research method $\mathcal{O}(N^{L_P+1}(N^2\!+\!NM))$, and the complexity of the element-wise BCD method is the lowest. 

\section{Simulation Results and Analysis}
In this section, simulation results are presented to validate the secrecy rate performance in mmWave and THz bands. 
All results are obtained by averaging over $1000$ independent trials.
We consider a scenario where the BS employs a uniform linear array (ULA) and the IRS is a uniform rectangular array (URA). 
Unless otherwise specified, the transmitter frequency is $f=0.3$~THz, and the transmit power is $P_s=25$~dBm, $M=16$. Due to a full-connected hybrid beamforming structure at BS, the RF chain number is less than $M$, $M_{RF}=10$.
The number of discrete values in $\mathcal{F}$ is $L_P=2^3$, and $\Delta \theta=\frac{\pi}{4}$. 
Additionally, the antenna gain is set to $12$~dBi. 
The complex channel gain $\alpha_B$ and $\alpha_{k,i}$ can be obtained on the basis of \cite{Bar2017TVT}.     

\begin{figure}[!t]
\centering
\subfigure[$R_s$ versus $L_P$.]{
\begin{minipage}{0.2\textwidth}
\centering
\includegraphics[width=1\textwidth]{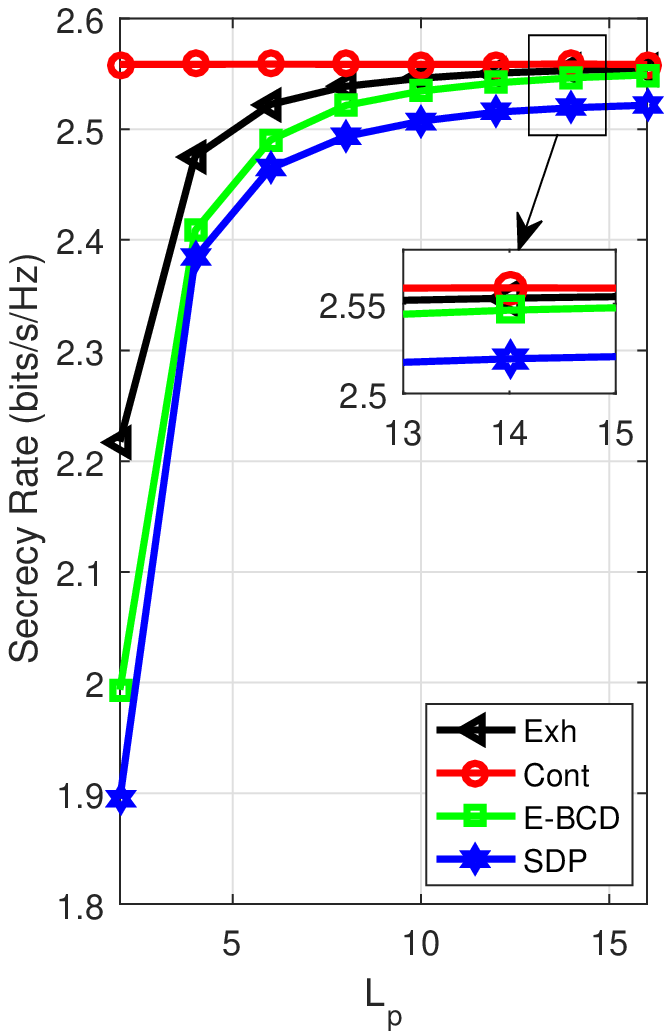}\\
\label{fig_2}
\end{minipage}}
\subfigure[$R_s$ versus $P_s$.]{
\begin{minipage}{0.2\textwidth}
\centering
\includegraphics[width=1\textwidth]{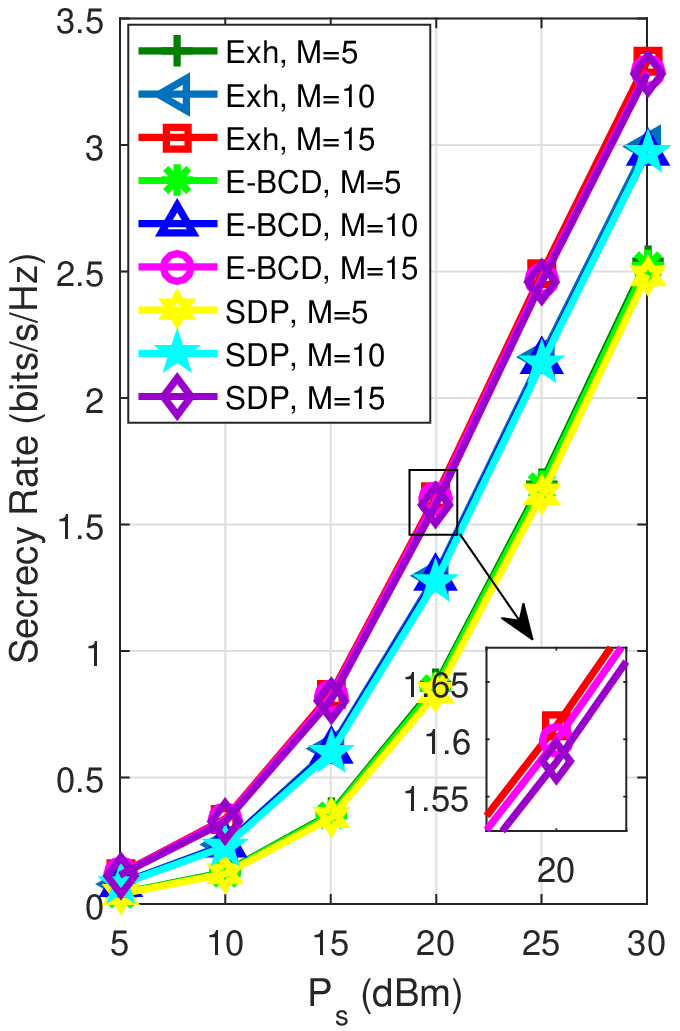}\\
\label{fig_3}
\end{minipage}}
\caption{Secrecy rate performance for IRS-assisted mmWave and THz system, in which $N=4$, the BS-to-IRS and IRS-to-Bob distances are $d_{sr}=d_{rd}=5$~m, and Eve is located near Bob with IRS-Eve distance $d_{re}=5$~m.}
\end{figure}
The secrecy rate with different number of discrete values of set $\mathcal{F}$ is shown in \figref{fig_2}. With the increase of $L_P$, the secrecy rate of proposed methods approaches that of exhaustive search method  (`Exh')\footnote{The solution of exhaustive search method is the optimal solution of discrete case, since all feasible solutions of phase shifts are searched.}.
This is rather intuitive, since an increased number of discrete values leads to a reduction in quantization error in (\ref{eq_12}), thereby enhancing secrecy rate. 
Besides, it is obvious that discrete schemes are upper-bounded by the continuous solution (`Cont'). 

The secrecy rate versus transmit power $P_s$ is shown in \figref{fig_3}. 
Based on the optimal beamformer obtained from Proposition \ref{prop_1}, it is easily proved that the secrecy rate is an increasing function of $P_s$. 
Thus as $P_s$ increase, 
the secrecy rate increases monotonously. Similarly as \figref{fig_2}, \figref{fig_3} reveals that our proposed methods can achieve near optimal performance as exhaustive search method. 
Besides, results also demonstrate that the more antennas are equipped at BS, the higher secrecy rate is achieved. 

The effect of the number of reflecting elements on the secrecy rate of IRS-based and BS-based interception is also investigated.\footnote{Here, optimal $(\mathbf{w}, \mathbf{\Theta})$ for BS-based interception can be derived iteratively.}.  
As shown in \figref{fig_4}, as N increases from $10$ to $100$, the secrecy rate increases monotonously. This is because the more reflecting elements result in sharper reflecting beams, thereby enhancing information security.
In particular, when Eve is located within the reflecting/transmit beam of IRS/BS, 
we assume that it can intercept and block $\rho$ portion of confidential signals.
Intuitively, the more information is blocked, the worse secrecy rate can be achieved. 
Thus, compared with interception without blocking, this case is more serious for mmWave and THz communications. 
However, since the IRS-assisted secure transmission scheme is designed, the secrecy rate is significantly improved compared with secure oblivious approach (only maximizing information rate at the legal user). 

\begin{figure}[!t]
\centering
\subfigure[Eve intercepts IRS.]{
\begin{minipage}{0.2\textwidth}
\centering
\includegraphics[width=1\textwidth]{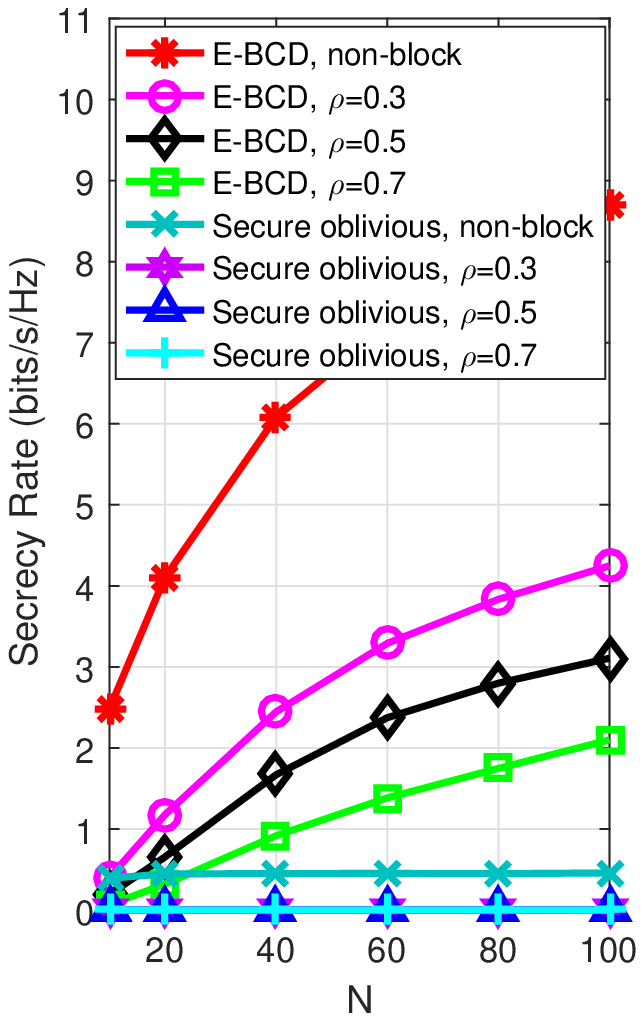}\\
\label{fig_4a}
\end{minipage}}
\subfigure[Eve intercepts BS.]{
\begin{minipage}{0.2\textwidth}
\centering
\includegraphics[width=1\textwidth]{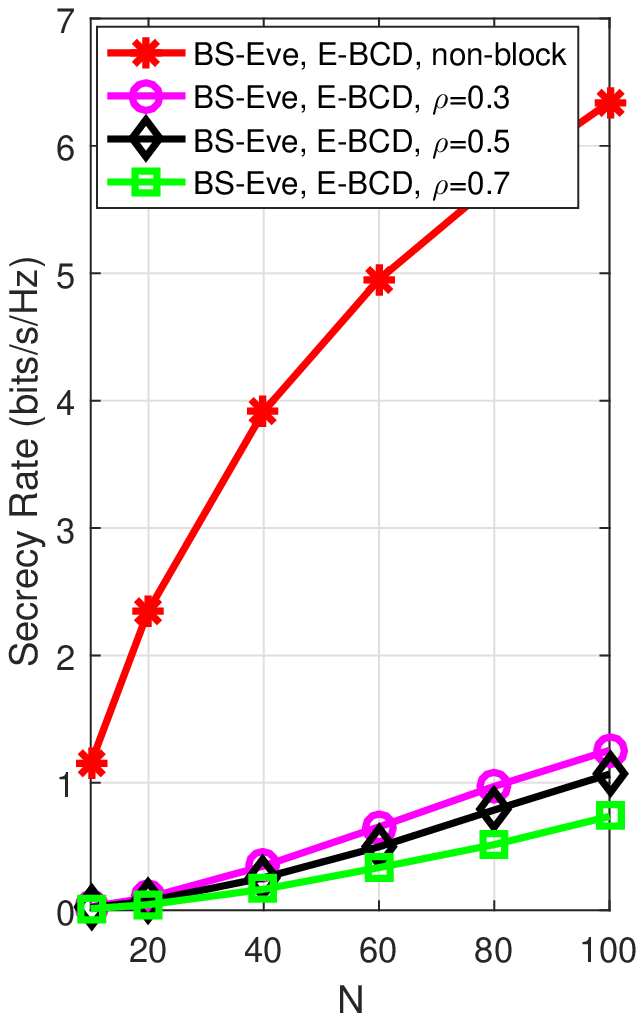}\\
\label{fig_4b}
\end{minipage}}
\caption{Secrecy rate versus the number of reflecting elements from $10$ to $100$. 
(a) Eve intercepts IRS, $d_{re}\!=\!5$~m for non-blocking, $d_{re}\!=\!2$~m for blocking, (b) Eve intercepts BS, $d_{se}\!=\!5$~m for non-blocking, $d_{se}\!=\!2$~m for blocking. 
}
\label{fig_4}
\end{figure}

\section{Conclusion}
This letter investigated the secrecy performance of IRS-assisted mmWave/THz systems.
Considering the hardware limitation at the IRS, the transmit beamforming at the BS and discrete phase shifts at the IRS have been jointly designed to maximize the system secrecy rate. 
To deal with the formulated non-convex problem, the original problem was divided into two subproblems under the rank-one channel assumption. 
Then the closed-form beamforming solution was derived,
and the reflecting matrix was obtained by the proposed SDP-based method and element-wise BCD method. 
Simulations demonstrated that our proposed methods can achieve the near optimal secrecy performance, and can combat eavesdropping  occurring at the BS and the IRS. 
 
\appendices
\section{Proof of Proposition 1}
In $\mathrm{P2.1}$, the beamforming vector $\mathbf{w}$ is always coupled with $\mathbf{b}^{H}$, thus the secrecy rate can be seen as a function of $|\mathbf{b}^{H}\!\mathbf{w}|^2$.  
Under the positive secrecy rate constraint, i.e., $\scriptsize{|\mathbf{g}_{D}^{H}\!\mathbf{\Theta}\mathbf{a}|^2\!/\sigma_D^2\!\!>\!\!|\mathbf{g}_{E}^{H}\!\mathbf{\Theta}\mathbf{a}|^2\!/\sigma_E^2}$, 
 it is easily proved that $R_s$ is an increasing function of $|\mathbf{b}^{H}\!\mathbf{w}|^2$. 
That is, $\mathrm{P2.1}$ is equivalent to $\max|\mathbf{b}^{H}\mathbf{w}|^2$ and is intuitively independent of $\mathbf{\Theta}$.

\section{Proof of Proposition 2}
Choosing each element at the IRS $\theta_i$ as one block of BCD, the objective function of $\mathrm{P2.2}$ can be reformulated as 
\begin{equation}
f(\theta_i)=\frac{c_{D,i}+d_{D,i}\cos(\theta_i+p_{D,i})}{c_{E,i}+d_{E,i}\cos(\theta_i+p_{E,i})},
\end{equation}
in which all parameters can be found in section III-B-2, and $c_{D,i},c_{E,i}\geq 1$, $c_{D,i}>d_{D,i}>0$, $c_{E,i}>d_{E,i}>0$. 
The sign of the derivative of $f(\theta_i)$ is determined by\footnote{To reduce complexity, $\sin(x)$-based expression is used instead of $\cos(x)$.}
\begin{equation}
h(x)=\sqrt{A_i^2+B_i^2}\sin(x)+C_i,
\end{equation} 
where $x\!=\!\theta_i\!+\!\phi$,  
$A_{i}\!=\!c_{D,i}d_{E,i}\cos(p_{E,i})\!-\!c_{E,i}d_{D,i}\cos(p_{D,i})$, $B_{i}\!= \!c_{D,i}d_{E,i}\sin(p_{E,i})\!-c_{E,i}d_{D,i}\sin(p_{D,i})$, 
and $ C_i=d_{D,i}d_{E,i}\sin(p_{E,i}-p_{D,i})$,   $\sin(\phi)=B_i\!/\!\sqrt{A_i^2\!+\!B_i^2}$, $\cos(\phi)=A_i/\sqrt{A_i^2+B_i^2}$.   
Then the main problem in deriving the unique optimal solution is to determine the value of $\phi$ or $\phi^{'}$.
\begin{itemize}
\item For $A_i>0$, we have $\cos(\phi)>0$ and $\phi=\arctan(\frac{B_i}{A_i})$. Thus, $\theta_i^{opt}=\pi-\arctan(\frac{B_i}{A_i})-\arcsin(\frac{-C_i}{\sqrt{A_i^2+B_i^2}})$.

\item For $A_i\!<\!0$, we have $\cos(\phi)\!<\!0$ and $\phi=\pi+\arctan(\frac{B_i}{A_i})$. Thus $\theta_i^{opt}=-\arctan(\frac{B_i}{A_i})-\arcsin(\frac{-C_i}{\sqrt{A_i^2+B_i^2}})$.
\end{itemize}

\vspace{-0.2cm}
\scriptsize
\bibliographystyle{IEEEtran}

\bibliography{IEEEabrv,THz_IRS}
\end{document}